\documentclass[american,aps,pra,reprint,superscriptaddress]{revtex4-1}
\usepackage[unicode=true,pdfusetitle, bookmarks=true,bookmarksnumbered=false,bookmarksopen=false, breaklinks=false,pdfborder={0 0 0},backref=false,colorlinks=false] {hyperref}
\hypersetup{ colorlinks,linkcolor=myurlcolor,citecolor=myurlcolor,urlcolor=myurlcolor}
\usepackage{braket,colortbl,cleveref,amsthm,amsmath,amssymb,txfonts}
\definecolor{myurlcolor}{rgb}{0,0,0.7}

\newcommand{\ketbra}[2]{|{#1} \rangle \langle {#2} |}

\theoremstyle{plain}
\newtheorem{thm}{Theorem}

\begin{document}

\title{Energy Cost of Creating Quantum Coherence}
\author{Avijit Misra}
\email{avijit@hri.res.in}
\affiliation{Harish-Chandra Research Institute, Allahabad, India}

\author{Uttam Singh}
\email{uttamsingh@hri.res.in}
\affiliation{Harish-Chandra Research Institute, Allahabad, India}

\author{Samyadeb Bhattacharya}
\email{samyadebbhattacharya@hri.res.in}
\affiliation{Harish-Chandra Research Institute, Allahabad, India}

\author{Arun Kumar Pati}
\email{akpati@hri.res.in}
\affiliation{Harish-Chandra Research Institute, Allahabad, India}

\begin{abstract}
We consider the physical situations where the resource theories of coherence and thermodynamics play competing roles. In particular, we study the creation of quantum coherence using unitary operations with limited thermodynamic resources. We find the maximal coherence that can be created under unitary operations starting from a thermal  state and find explicitly the unitary transformation that creates the maximal coherence. Since coherence is created by unitary operations starting from a thermal state, it requires some amount of energy. This motivates us to explore the trade-off between the amount of coherence that can be created and the energy cost of the unitary process. We also find the maximal achievable coherence under the constraint on the available energy. Additionally, we compare the maximal coherence and the maximal total correlation that can be created under unitary transformations with the same available energy at our disposal. We find that when maximal coherence is created with limited energy, 
the total correlation created in the process is upper bounded by the maximal coherence and vice versa. For two qubit systems we show that there does not exist any unitary transformation that creates maximal coherence and maximal total correlation simultaneously with a limited energy cost. 
\end{abstract}

\maketitle
\section{Introduction}

The superposition principle in quantum physics gives rise to the classically counter-intuitive traits like coherence and entanglement \cite{entanglement-rmp}. Over the past few years, several works have been done starting from quantifying quantum superposition \cite{Aberg2006} to establishing the full fledged resource theories of coherence \cite{Baumgratz2014, Marvian14, Gour2008}. The restrictions underlying a quantum resource theory (QRT) are manifestations of the physical restrictions that govern the physical processes. For example, in the QRT of entanglement, the consideration of local operations 
and classical communication as the allowed operations stems from the natural limitation to implement global operations on a multipartite quantum system with parties separated apart. Considering the technological advancements towards processing of small scale quantum systems and proposals of nano-scale heat engines, it is of utmost importance to investigate the thermodynamic perspectives of quantum features like coherence and entanglement and it has attracted a great deal of interest (see 
for example Refs. \cite{Alicki2013, Hovhannisyan2013, Brunner2014, Abah2014, Correa2014, Gallego2014, Galve2014, Carlisle2014, Scarani2002, Eisert2004, Huber2014, Bruschi2015, Anders-review,Huber-review,Anders2015, Llobet2015, Korzekwa2016}). Since energy conservation restricts 
the thermodynamic processing of coherence, a quantum state having coherence can be viewed as a resource in thermodynamics as it allows transformations that are otherwise impossible  \cite{Rudolph114, Rudolph214}. These explorations have got renewed interests in recent times for their possible implications in various areas like information theory \cite{Maruyama2009,Rio2011,  Fernando2013, HorodeckiI2013, Brandao2015} and quantum biology \cite{Engel2007, Lloyd2011, Collini2010, Lambert2013, Chin2012}. In particular, it is also shown that quantum coherence allows for better transient cooling in the absorption refrigerators and this phenomenon is dubbed as coherence-assisted single shot cooling \cite{Mitchison2015}.

In this work, with aforesaid motivation, we explore the intimate connections between the resource theory of quantum coherence and thermodynamic limitations on the processing of quantum coherence. In particular, we  study the creation of quantum coherence by unitary transformations with limited energy. We go on even further to present a comparative investigation of creation of quantum coherence and mutual information within the imposed thermodynamic constraints. Considering a thermally isolated quantum system initially in a thermal state, we perform an arbitrary unitary operation 
on the system to create coherence. First we find the upper bound on the coherence that can be created using arbitrary unitary operations starting from a fixed thermal state and then we show explicitly that irrespective of the temperature of the initial thermal state the upper bound on coherence can always be saturated. Such a physical process will cost us some amount of energy and hence it is natural to ask that if we have a limited supply of energy to invest then what is the maximal achievable coherence in such situations? Further, we investigate whether both coherence and mutual information can be created maximally by applying a single unitary operation on a two qubit quantum system. We find that it is not possible to achieve maximal quantum coherence and mutual information simultaneously. Our results are relevant for the quantum information processing in physical systems 
where thermodynamic considerations cannot be ignored as we have discussed in the preceding paragraph.

The organization of the paper is as follows. In Sec. \ref{sec1} we introduce the concepts that are necessary for further expositions of our work and present the form of the unitary operation that saturates the upper bound on the amount of coherence that can be created by applying unitary transformations starting from an incoherent state. In Sec. \ref{sec2} we provide various results on the creation of coherence when we have a limited availability of energy. In Sec. \ref{sec3} we compare the processes of creating maximal coherence and maximal mutual information again with the limited thermodynamic resource. Finally we conclude with some possible implications in Sec. \ref{sec4}. 

\section{Maximum Achievable Coherence under arbitrary unitary operations}
\label{sec1}
In this section, we first discuss the QRT of coherence and then find a protocol for achieving maximal quantum coherence under unitary operations.
\subsection{Quantum Coherence}
Any QRT is formed by identifying the relevant physical restrictions on the set of quantum operations and preparation of quantum states. For example, in the QRT of thermodynamics the set of allowed operations is identified as the thermal operations and the only free state is the thermal state for a given fixed Hamiltonian \cite{Fernando2013, HorodeckiI2013}. The QRT of thermodynamics is well accepted and has seen enormous progress in recent years \cite{Brandao2015, Rudolph114, Rudolph214, Narasimhachar2015}. However, there is still an ongoing debate on the possible choices of restricted operations that will define the resource theory of coherence for finite dimensional quantum systems 
\cite{Marvian14, Baumgratz2014, StreltsovA2015, Yadin2015}. The field of QRT of coherence has been significantly advanced over past few years \cite{Streltsov2015, Manab2015, Angelo2015, Uttam2015, Fan15, Pinto2015, Du2015, Yao2015, Xi2015, Girolami14, Bromley2015, Killoran2015, ZhangA2015, UttamA2015, Cheng2015, Mondal2015, MondalA2015, Kumar2015, 
Mani2015, Bu2015, Chitambar2015, ChitambarA2015, BuA2015, StreltsovA2015, StreltsovAA2015, ChitambarA2016, Marvian2016, Peng2016}. The measures of coherence are inherently basis dependent and the relevant reference basis is provided by the experimental situation at 
hand. Here we will be contended with the measures of coherence that are obtained using QRT of coherence based on the incoherent operations as introduced in Ref. \cite{Baumgratz2014}. 
However, very recently a refinement over Ref. \cite{Baumgratz2014} on the properties that a coherence measure should satisfy has been proposed in Ref. \cite{Peng2016}. This refinement imposes an extra condition on the measures of coherence such that the set of states having maximal coherence value with respect to the coherence measure and the set of maximally coherent states, as defined in Ref. \cite{Baumgratz2014}, should be identical.
%
%
%
%
Moreover, the unitary freedom in Kraus decomposition of a quantum channel implies that an incoherent channel with respect to a particular physical realization (Karus decomposition) may not be incoherent with respect to other realizations of the same channel. This has led to the introduction of genuinely incoherent operations \cite{StreltsovA2015}. Genuinely incoherent operations are the ones that preserve the incoherent states. It is proved that the set of the genuinely incoherent operations is a strict subset of the incoherent operations, therefore, every coherence monotone based on the set of the incoherent operations is also a genuine coherence monotone \cite{StreltsovA2015}.

In this work, we consider the relative entropy of coherence as a measure of coherence which enjoys various operational interpretations \cite{Winter2015, UttamA2015}. The relative entropy of coherence of any state $\rho$ is defined as $C_{r}(\rho) = S(\rho^d) - S(\rho)$ where $S(\rho)=-\mathrm{Tr} (\rho \log \rho)$, is the von Neumann entropy and $\rho^d=\sum_i\langle i | \rho |i \rangle | i 
\rangle\langle i|$ is the diagonal part of $\rho$ in the reference basis $\{\ket{i}\}$. Here and in the rest parts of the paper all the logarithms are taken with respect to base $2$. It is to be noted that the relative entropy of coherence is also a genuine coherence monotone. Moreover, it also satisfies the additional requirement as proposed in Ref. \cite{Peng2016}.
\subsection{Maximum Achievable Coherence}
Let us now consider the creation of maximal coherence, which we define shortly, starting from a thermal state by unitary operations. The prime motivation for starting with a thermal initial state is that the surrounding may be considered as a thermal bath and as the system interacts with the 
surrounding, it eventually gets thermalized. However, our protocol for creating maximal coherence is applicable for any incoherent state. Let us now consider an arbitrary quantum system in contact with 
a heat bath at temperature $ T=1/\beta$ ( we set the Boltzmann constant to be unity and follow this convention throughout the paper). The thermal state of a system with Hamiltonian $H = \sum_{j=1}^{d}E_j\ket{j}\bra{j}$ is given by
\begin{align}
 \rho_T = \frac{1}{Z}e^{-\beta H},
\end{align}
where $d$ is the dimension of the Hilbert space and $Z = \mathrm{Tr}[e^{-\beta H}]$ is the partition function. The maximum amount of coherence $C_{r,max}(\rho_f)$ that can be created starting from $\rho_T$ by unitary 
operations is given by
\begin{align}
 \label{eq:max}
 C_{r,max}(\rho_f)&= 
\max_{\{\rho_f|S(\rho_f)=S(\rho_T)\}}\{S(\rho_f^D)-S(\rho_T)\}.
\end{align}
As the maximum entropy of a quantum state in $d$-dimension is $\log d$, the amount of coherence that can be created starting from $\rho_T$, by a unitary transformation, always follows the inequality
\begin{align}
\label{Eq:bound}
 C_{r}(\rho_f) \leq \log d - S(\rho_T).
\end{align} 
Now the question is whether the bound is tight or not, i.e., is there any unitary operation that can lead to the creation of $\log d-S(\rho_T)$ amount of coherence starting from $\rho_T$? But before we answer this question, let us digress on its importance first. Coherence in energy eigenbasis plays a crucial role in quantum thermodynamic protocols and several quantum information processing tasks. For example in Ref. \cite{Mitchison2015}, it has been demonstrated that if the initial qubits of a three qubit refrigerator possess even a little amount of coherence in energy eigenbasis then the cooling can be significantly better. In the small scale refrigerators, the three constituent qubits initially remain in corresponding thermal states associated with the three thermal baths. Therefore, one needs to create 
coherence by external means. Hence, creation of coherence from thermal states may be fruitful and far-reaching for better functioning of various nano scale thermal machines and diverse thermodynamic protocols. These are the main motivations for studying creation of coherence from thermodynamic perspective. We consider closed quantum systems and hence allow only unitary operations for creating coherence. Of course, after creating the coherence via the unitary transformation we have to isolate or take away the system from the heat bath so that it does not get thermalized again. Now, we show that the bound in Eq. \eqref{Eq:bound} is achievable by finding the unitary operation $U$ such that $\rho_f = U\rho_iU^\dag$ has maximal amount of coherence. Since the relative entropy 
of coherence of $\rho_f$ is given by $S(\rho_f^D)-S(\rho_f)$, one has to maximize the entropy of the  diagonal density matrix $\rho_f^D$. The quantum state that is the diagonal of a quantum state $\rho$ is denoted as $\rho^D$ throughout the paper.

First, we construct a unitary transformation that results in rotating the energy eigenbasis to the maximally coherent basis as follows. The maximally coherent basis $\{\ket{\phi_j}\}_j$ is defined as $\ket{\phi_j} = \mathbb{Z}^j\ket{\phi}$, where
\begin{align}
 \mathbb{Z} = \sum_{m=0}^{d-1} e^{\frac{2\pi i m}{d}}\ket{m}\bra{m},
\end{align}
and $\ket{\phi}=\frac{1}{\sqrt{d}}\sum_{i=0}^{d-1}\ket{i}$. It can be verified easily that, $\langle \phi_j\ket{\phi_k}=\delta_{jk}$. Also, note that all the states in $\{\ket{\phi_j}\}_j$ and $\ket{\phi}$ are the maximally coherent states \cite{Baumgratz2014} and have equal amount of the relative entropy of coherence which is equal to $\log d$. Now consider the unitary operation
\begin{align}
\label{eq:choice-u}
 U = \sum_{j}\ket{\phi_j}\bra{j},
\end{align}
which changes energy eigenstate $\ket{j}$ to the maximally coherent state $\ket{\phi_j}$. Starting from the thermal state $\rho_T$, the final state $\rho_f$ after the application of $U$ is given by
\begin{align}
 \rho_f = \sum_{j} \frac{e^{-\beta E_j}}{Z} \ket{\phi_j}\bra{\phi_j}.
\end{align}
Since $\rho_f$ is a mixture of pure states that all have maximally mixed diagonals, the bound in Eq. (\ref{Eq:bound}) is achieved. We note that $U$ in Eq. \eqref{eq:choice-u} is only one possible choice among the possible unitaries achieving the bound in Eq. (\ref{Eq:bound}). For example, any permutation of the indices $j$ of $\ket{\phi_j}$ in Eq. \eqref{eq:choice-u} is also a valid choice to achieve the bound. It is worth mentioning that even though we consider thermal density matrix to start with to create maximal coherence in energy eigenbasis, following the same protocol maximal coherence can be created from any arbitrary incoherent state in any arbitrary reference basis.

To create coherence by unitary operations starting from a thermal state, some amount of energy is required. Now, let us ask how much energy is needed on an average to create the maximal amount of coherence. Let $\rho_T \rightarrow V\rho_{T} V^\dag$, then the energy cost of any arbitrary unitary operation $V$ acting on the thermal state is given by
\begin{align}
\label{Eq:work-cost}
 W =\mathrm{Tr}[H (V\rho_{T} V^\dag - \rho_{T})].
\end{align}
Since we are dealing with energy eigenbasis, we have $E(\rho^D)=E(\rho)$. Here $E(\rho)=\mathrm{Tr}(H \rho)$ is the average energy of the system in the state $\rho$. The energy cost to create maximum coherence starting from the thermal state $\rho_T$ is given by
\begin{align}
 W_{max} =\mathrm{Tr}[H (U\rho_{T} U^\dag - \rho_{T})]= \frac{1}{d}\mathrm{Tr}[H] -  \frac{1}{Z}\mathrm{Tr}[H e^{-\beta H}].
\end{align}
Here $U$ is given by Eq. \eqref{eq:choice-u}. Note that maximal coherence can always be created by unitary operations starting from a finite dimensional thermal state at an arbitrary finite temperature, with finite energy cost. However, it is not possible to create coherence by unitary operation starting from a thermal state at infinite temperature i.e., the maximally mixed state.
\section{Creating Coherence with limited energy}
\label{sec2}
Since energy is an independent resource, it is natural to consider a scenario where creation of coherence is limited by a constraint on available energy. In this section we consider creation of optimal amount of coherence at a limited energy cost $\Delta E$ starting from $\rho_T$. To maximize the coherence, one needs to find a final state $\rho_f$ whose diagonal part $\rho_f^D$ has maximum entropy with fixed average energy $E_T+\Delta E$, where $E_T$ is the average energy of the initial thermal state $\rho_T$. Note that $E(\rho_f^D)=E(\rho_f)$. From maximum entropy principle  
\cite{JaynesA1957, JaynesB1957}, we know that the thermal state has maximum entropy among all states with fixed average energy. Therefore, the maximum coherence $C^{\Delta E}_{r,max}$, that can be created with $\Delta E$ amount of available energy is upper bounded by

\begin{align}
\label{Eq:bound-energy}
 C^{\Delta E}_{r,max}\leq S(\rho_{T'}) - S(\rho_T).
\end{align}
Here $\rho_{T'}$ is a thermal state at a higher temperature $T'$ such that $\Delta E = \mathrm{Tr}[H(\rho_{T'}- \rho_{T})]$. Thus, in order to create maximal coherence at a limited energy cost, one should look for a protocol such that the diagonal part of $\rho_f$ is thermal state at a higher temperature $T'$ (depending on the energy spent $\Delta E$), i.e., $ \rho_{f}^D =  \rho_{T'}$. Now it is obvious to inquire whether there always exists an optimal unitary $U^{*}$ that serves the purpose. Theorem \ref{prop1} answers this question in affirmative. 
\begin{thm}
\label{prop1}
There always exists a real orthogonal transformation $R$ that creates maximum coherence $S(\rho_{T'}) - S(\rho_T)$, starting from the thermal state $\rho_T$ and spending only $\Delta E = \mathrm{Tr}[H(\rho_{T'}- \rho_{T})]$ amount of energy.
\end{thm}

\begin{proof}
To prove the theorem, we first show that the unitary transformations on a quantum state induce doubly stochastic \footnote{A $d\times d$ matrix $M=(M_{ab})$ ($a\in \{1,\ldots,d\}$, $b\in \{1,\ldots,d\}$ and $M_{ab}\geq 0$) is called a doubly stochastic matrix if $\sum_aM_{ab}=1$ for all $b$ and $\sum_bM_{ab}=1$ for all $a$ \cite{Bhatia1997}} maps on the diagonal part of the quantum state. Note that, we start from the thermal state $\rho_T=\sum_{j=0}^{d-1}\frac{e^{-\beta E_j}}{Z}\ket{j}\bra{j}$. The diagonal part of $\rho_T$ transforms under the action of a unitary $U$ as follows:
\begin{align}
 \mbox{diag}\{U \rho_T U^\dag\} =\sum_{i=0}^{d-1} q_i\ket{i}\bra{i},
\end{align}
where $q_i =\frac1Z \sum_{j=0}^{d-1}M_{ij}e^{-\beta E_j}$ and $M$, with entries $M_{ij}~=~ 
\bra{i}U\ket{j}\bra{j}U^\dag\ket{i}$, is a doubly stochastic matrix. Therefore, the diagonal part is transformed by the doubly-stochastic matrix $M$ such that
\begin{align}
\vec{ Q}= M \vec{P_T},
\end{align}
where $\vec{P_T} =\frac{1}{Z} \{e^{-\beta E_0}, e^{-\beta E_1},\dots,e^{-\beta E_{d-1}}\}^T$ is the diagonal vector corresponding to the initial thermal state $\rho_T$ and $\vec{ Q}$ is the diagonal vector corresponding to the diagonal part of the final state. For two thermal states $\rho_{T'}$ and $\rho_T$, corresponding to the same Hamiltonian, we have $\vec{P_{T'}}\prec \vec{P_{T}}$ if $T'>T$ \cite{Canosa2005}. From the results of the theory of majorization \cite{Bhatia1997}, it can be concluded that there always exists an orthostochastic \footnote{Consider a $d\times d$ real orthogonal matrix $O=(O_{ab})$ ($a\in \{1,\ldots,d\}$, $b\in \{1,\ldots,d\}$ and $O_{ab}\geq 0$) such that $O^TO=\mathbb{I}$. Here superscript $T$ denotes the transpose and $\mathbb{I}$ is a $d\times d$ identity matrix. The matrix $(O_{ab}^2)$ forms a doubly stochastic matrix and such a matrix is called orthostochastic matrix \cite{Bhatia1997}.} matrix $B$ such that $\vec{P_{T'}}=B\vec{P_{T}}$. Hence, the real orthogonal 
operator $R$, corresponding to the orthostochastic matrix $B$, transform the initial thermal state $\rho_T$ to a final state $\rho_f$ such that $\rho_f^D=\rho_{T'}$. Therefore, there always exists a real orthogonal transformation $R$ that creates $S(\rho_{T'}) - S(\rho_T)$ amount of coherence, starting from the thermal state $\rho_T$ and spending only $\Delta E = \mathrm{Tr}[H(\rho_{T'}- \rho_{T})]$ amount of energy. This completes the proof.
\end{proof}

\subsection{Qubit system}
In the following, we find out explicitly the real unitary transformation that allows creation of maximal coherence with limited energy at our disposal for the case of a qubit system with the Hamiltonian $H= E\ketbra{1}{1}$. The initial thermal state is given by $\rho_T=p\ket{0}\bra{0}+(1-p)\ket{1}\bra{1}$ 
with $p= \frac{1}{1+e^{-\beta E}}$. Now our goal is to create maximal coherence by applying an optimal unitary $U^*$, investing only $\Delta E$ amount of energy. The average energy of the initial thermal state $\rho_T$ is given by $(1-p)E$. As we have discussed earlier that for maximal coherence creation with $\Delta E$ energy cost, the diagonal part of the final state must have to be a thermal state, $\rho_{T'}= q\ket{0}\bra{0}+(1-q)\ket{1}\bra{1}$, at some higher temperature $T'$, with average energy $(1-p)E+\Delta E$. Here, $q$, and hence $T'$, is determined from the energy constraint as $q=p-\frac{\Delta E}{E}= \frac{1}{1+e^{-\beta' E}}$. From theorem \ref{prop1}, it is evident that there always exists a rotation operator $R$ which creates the maximal coherence. Consider a rotation operator of the form 
\begin{align}
\label{rot-mat}
R(\theta)= \left (  \begin{array}{cc}
                 \cos\theta & -\sin\theta\\
                 \sin\theta &  \cos\theta
                 \end{array} \right)
\end{align}
that transforms $\rho_{T}$ as follows
\begin{align}
\rho_f&=R(\theta)\rho_T R^{T}(\theta)\nonumber\\
&= \left (  \begin{array}{cc}
                 p\cos^2\theta +(1-p)\sin^2\theta & (2p-1)\sin\theta\cos\theta\\
                 (2p-1)\sin\theta\cos\theta &  p\sin^2\theta +(1-p)\cos^2\theta
                 \end{array} \right).\nonumber
\end{align}
We need the diagonal part of the final state to be the thermal state $\rho_{T'}$ at a higher temperature $T'$. Therefore,  
\begin{align}
\label{qvalue}
\begin{array}{ll}
~~~~q~~~=p\cos^2\theta+(1-p)\sin^2\theta.
\end{array}
\end{align}
As R.H.S of Eq. (\ref{qvalue}) is a convex combination of $p$ and $(1-p)$ and $p\geq q\geq 1/2\geq (1-p)$, by suitably choosing $\theta$ we can reach to the desired final state $\rho_f$ such that $\rho_f^D=\rho_{T'}$. The angle of rotation $\theta$ is given by
\begin{equation}
\label{theta}
 \theta= \cos^{-1}\left(\sqrt{\frac{p+q-1}{2p-1}}\right).
\end{equation}
Thus, the maximal coherence at constrained energy cost $\Delta E$ can be created from a qubit thermal state by a two dimensional rotation operator as given by Eq. (\ref{rot-mat}). 

\subsection{Qutrit system}
For qubit systems, a two dimensional rotation with the suitably chosen $\theta$ is required to create maximum coherence starting from a thermal state at a finite temperature with limited available energy. For higher dimensional systems, it follows from theorem 1, that there always exists a rotation which serves the purpose of maximal coherence creation. However, finding the exact rotation operator for a given initial thermal density matrix and energy constraint is not an easy task. Even for a qutrit system finding the optimal rotation is nontrivial. In what follows, we demonstrate the protocol for creating maximal coherence with energy constraint starting from a thermal state for qutrit systems. Note that by applying a unitary operation on a thermal qubit, one has to invest some energy and thus, the excited state population corresponding to diagonal part of the final qubit is always increased. Therefore, for the case of qubit systems, one only has to give a rotation by an angle $\theta$, depending on the available energy to create maximal coherence starting from a given thermal state.
 For a thermal state in higher dimension, we know that with the increment in temperature (energy), the occupation probability of 
the ground state will always decrease and the occupation probability will increase for the highest 
excited state. But what will happen for the intermediate energy levels? Let us first answer this particular question considering an initial thermal state of the form 
\begin{align}
\label{3q-1}
\rho_T = \sum_{j=1}^{d} p_j\ketbra{j}{j},
\end{align}
where $p_j=\frac{e^{-E_j/T}}{\sum_j e^{-E_j/T}}$ is the occupation probability of the $j^{\mathrm{th}}$ energy level. Differentiating $p_j$ with respect to the temperature, we get
\begin{align}
\label{3q-3}
\frac{\partial p_j}{\partial T}=-\frac{(\langle E\rangle_T-E_j)}{T^2} p_j.
\end{align}
Therefore, for energy levels lying below the average energy of the thermal state, the occupation probabilities will decrease with the increase of temperature and the occupation probabilities will increase for the energy levels lying above the average energy. Making use of this change in occupation probabilities, we now provide a protocol for maximum coherence creation in thermal qutrit systems with a constraint on the available energy. We consider a qutrit system with the system Hamiltonian $H= E\ketbra{1}{1}+2E\ketbra{2}{2}$. The initial thermal qutrit state is given by $\rho_T=p\ketbra{0}{0}+(1-p-q)\ketbra{1}{1}+q\ketbra{2}{2}$ with average energy $\langle E\rangle_T=(1-p-q)E+2qE$.  Here $p=1/Z$ and $q=e^{-2\beta E}/Z$, where $Z=1+e^{-\beta E}+e^{-2\beta E}$ is the partition function. The diagonal density matrix of the final state is a thermal qutrit state at temperature  $T'$ with average energy $\langle E\rangle_T'=(1-p-q)E+2qE+\Delta E$, when we create coherence with $\Delta E$ energy constraint.

We show that just two successive rotations in two dimension suffice the purpose of maximum amount of coherence creation. For equal energy spacing of $\{0,E,2E\}$, the average energy at infinite temperature is given by $\langle E\rangle_{\infty}=E$. So, for an arbitrary finite temperature, the 
condition $E>\langle E\rangle_T$ holds true. Thus for the aforementioned qutrit thermal system, with the increase in temperature, the occupation probabilities of the first and second excited states will 
always increase at the expense of the decrease in occupation probability for the ground state. The diagonal elements of the final state should be the occupation probabilities of the thermal state at a higher temperature $T'$, given by $p', 1-p'-q'~ \mbox{and}~q'$ for the ground, first and second excited states, respectively. From the conservation of probabilities, it follows that
\begin{align}
\label{3q-4}
-\Delta p&= p-p'= (q'-q)+(1-p'-q')-(1-p-q)\nonumber\\
&=\Delta q+\Delta(1-p-q),
\end{align}
Note that, we always have $-\Delta p>\Delta q>0$.

Now, let us first apply a rotation about $\ket{1}$. Physically, this rotation creates coherence between basis states $\ket{0}$ and $\ket{2}$.  The rotation can be expressed by the unitary $R_1(\alpha)=e^{-i\alpha J_{1}}$, where 

\begin{align}
\label{3q-5}
J_{1}=\left(\begin{array}{ccc}
             0&0&i\\
             0&0&0\\
             -i&0&0\\
\end{array} \right)            
\end{align}
is the generator of the rotation. Then another rotation is applied about $\ket{2}$, which is given by $R_2(\delta)=e^{-i\delta J_{2}}$, where 
\begin{align}
\label{3q-6}
J_{2}=\left(\begin{array}{ccc}
             0&-i&0\\
             i&0&0\\
             0&0&0\\
\end{array} \right).           
\end{align}
After the action of two successive rotations, given by $R_2(\delta)R_1(\alpha)$, we have
\begin{align}
\label{3q-8}
q'=q\cos^2\alpha+p\sin^2\alpha
\end{align}
and 
\begin{align}
\label{3q-9}
p'&=(q\sin^2\alpha+p\cos^2\alpha)\cos^2\delta+(1-p-q)\sin^2\delta\nonumber\\
&=(p-\Delta q)\cos^2\delta+(1-p-q)\sin^2\delta.
\end{align}
From Eq. (\ref{3q-8}), it is clear that $q'$ is a convex combination of $p$ and $q$, and since $q<q'<p$, there always exists a angle of rotation $\alpha$, depending on the available energy so that the protocol can be realized. The angle of rotation is given by $\alpha=\cos^{-1}\sqrt{\frac{p-q'}{p-q}}$, where $\alpha \in [0,\pi/2]$. Similarly, Eq. (\ref{3q-9}), suggests that $p'$ is a convex combination of $(p-\Delta q)$ and $(1-p-q)$ and since $1-p-q<p'<p-\Delta q$ (Eq.\ref{3q-4}), one can always achieve any desired value of $p'$, by suitably choosing $\delta \in [0,\pi/2]$, with $\delta=\cos^{-1}\sqrt{\frac{p'-(1-p-q)}{(p-\Delta q)-(1-p-q)}}$. Thus, maximal coherence at finite energy cost can be created by two successive two dimensional rotations starting from a thermal state of a qutrit system. Note that we have considered equal energy spacing $\{0,E,2E\}$, however, the above protocol will hold for any energy spacing for which the condition $E_1 > \langle E\rangle_T$ holds, where $E_1$ is the energy of the 
energy eigenstate $\ket{1}$.

\section{Coherence versus Correlation}
\label{sec3}
In this section we carry out a comparative study between maximal coherence creation and maximal total correlation creation (also see Ref. \cite{Huber2014}) with limited available energy. We consider an arbitrary $N$ party system acting on a Hilbert space $\mathcal{H}^{d_1}\otimes \mathcal{H}^{d_2}\otimes\ldots\otimes \mathcal{H}^{d_N}$. The Hamiltonian of the composite system is non-interacting and given by $H_{tot}=H_1+H_2+\ldots+H_N$. For the sake of simplicity we consider $H_1=H_2=\ldots=H_N=H$. However, our results hold in general. Suppose there exists an optimal unitary operator $U^*$ which creates maximal total correlation from initial thermal state $\rho_T$ with $\Delta E$ energy cost. It is shown in Ref. \cite{Huber2014} that the maximal correlation (multipartite mutual information) that can be created by a unitary transformation with energy cost $\Delta E$ is given by
\begin{equation}
\label{max-cor}
I^{\Delta E}_{max}= \sum_i\left[S(\rho^i_{T'})-S(\rho^i_{T})\right].
\end{equation}
In the protocol to achieve the maximal correlation, the subsystems of the composite system $\rho^N_T$ transform to the thermal states $\rho^i_{T'}$ of the corresponding individual systems at some higher temperature $T'$ \cite{Huber2014}. It is interesting to inquire that how much coherence is created 
during this process as in several quantum information processing tasks it may be needed to create both the coherence and correlation, simultaneously. The amount of coherence created $ C_r|_{I^{\Delta E}_{max}}$, when the unitary transformation creates maximal correlation is given by
\begin{equation}
 C_r|_{I^{\Delta E}_{max}}= S(\rho_f^D)-\sum_iS(\rho^i_{T}). 
\end{equation}
As the Hamiltonian is noninteracting, $\rho_f^D$ and the product of the marginals ($\prod^{\otimes i} \rho^i_{T'}$) have the same average energy. Since the product of the marginals is the thermal state of the composite system at temperature $T'$, the maximum entropy principle implies that $\sum_iS(\rho^i_{T'})\geq S(\rho_f^D)$. Hence, $C_r|_{I^{\Delta E}_{max}}\leq I^{\Delta E}_{max}$. Therefore, when one aims for maximal correlation creation the coherence created is always bounded by the amount of correlation created. Now, we ask the converse, i.e., how much correlation can be created when one creates maximal coherence by a unitary operation with the same energy constraint $\Delta E$? The maximal coherence that can be created in this scenario by unitary transformation 
with energy constraint is given by
\begin{equation}
\label{max-coh}
C^{\Delta E}_{r,max}= \sum_i\left[S(\rho^i_{T'})-S(\rho^i_{T})\right].
\end{equation}
Note that the maximal achievable coherence is equal to the maximal achievable correlation (cf. Eq. (\ref{max-cor})), but the protocols to achieve them are completely different. When the maximal coherence is created, the diagonal of the final density matrix is a thermal state at some higher temperature while the maximal correlation is created when the product of the marginals of the final 
state is a thermal state at some higher temperature. Therefore, when the maximum amount of coherence $C^{\Delta E}_{r,max}$  is created, the correlation $I|_{C^{\Delta E}_{r,max}}$ that is created simultaneously always satisfies
\begin{equation}
 I|_{C^{\Delta E}_{r,max}}\leq C^{\Delta E}_{r,max}.
\end{equation}
The above equation again follows from the maximum entropy principle and the fact that the diagonal part and the product of the marginals have same average energy. Therefore, when one aims for maximal coherence creation, the amount of correlation that can be created at the same time is always bounded by the maximal coherence created and vice versa.

It is also interesting to inquire whether one can create maximal coherence and correlation simultaneously. In the following, we partially answer this question. For two qubit systems we show that there does not exist any unitary which maximizes both the coherence and correlation, simultaneously. Let the Hamiltonian of the two qubit system be given by $H_{AB}=H_A+H_B$ with $H_A\neq H_B$, in general. Later, we also consider $H_A=H_B$. The initial state is the thermal state at temperature $T$ 
and given by
\begin{equation}
\rho_{AB,T}= \mbox{diag}\{pq,p(1-q),(1-p)q,(1-p)(1-q)\},
\end{equation}
where $p=1/(1+e^{-\beta E_{A}})$, $q=1/(1+e^{-\beta E_{B}})$, $H_A=E_A \ketbra{1}{1}$ and $H_B=E_B \ketbra{1}{1}$. Consider the protocol of Ref. \cite{Huber2014} to create the maximum correlation. In that scenario, the marginals are the thermal states at a higher temperature $T'$. Let the final state of the two qubit system after the unitary transformation is given by
\begin{align}
\rho_{AB}^f= \sum_{ijkl} a_{ijkl}|i\rangle\langle j|\otimes|k\rangle \langle l|. 
\end{align}
As the marginals are thermal, $a_{iikl}=0$ if $k\neq l$ and $a_{ijkk}=0$ if $i\neq j$. Thus, the maximally correlated state that is created by investing a limited amount of energy is an $X-$state \footnote{$X-$states are a special class of states that have been analyzed in great detail in context of analytical calculations of quantum discord \cite{Henderson2001, Ollivier2001} among others. The term $X-$states has been coined in Ref. \cite{Yu2007} for their visual appearance. For a bipartite qubit quantum systems, the states $ \rho_X$ of the form 
\begin{align*}
 \rho_X:=\left (  \begin{array}{cccc}
                 \rho_{00} &0 &0 & \rho_{03}\\
                0 & \rho_{11} & \rho_{12} & 0\\
                0 &  \rho_{21} & \rho_{22} &0\\
                 \rho_{30} &0 &0 & \rho_{33}
                 \end{array} \right)
\end{align*}
are called $X-$states. In general, any density matrix that has nonzero elements only at the diagonals and anti-diagonals is called an $X-$ state. For a detailed exposition of $X-$states see Ref. \cite{Rau2009}.}. While for maximal coherence creation, the diagonal part of the final state is a thermal state at higher temperature $T'$. Therefore, the diagonal part of the final state must be of the form
\begin{align}
 \rho^{fD}_{AB,T'}= \mbox{diag}\{p'q',p'(1-q'),(1-p')q',(1-p')(1-q')\},\nonumber
\end{align}
where $p'=1/(1+e^{-\beta' E_{A}})$, $q'=1/(1+e^{-\beta' E_{B}})$, $p'<p$ and $q'<q$ as $\beta'<\beta$. We show in \cref{{appendix-some-proof}}, separately for (i)$E_A=E_B$ and (ii)$E_A\neq E_B$, that there is no 
such unitary transformation which serves the purpose. It will be interesting to explore what happens for higher dimensional systems. 

\section{Conclusion}
\label{sec4}
In this article, we have studied the creation of quantum coherence by unitary transformations starting from a thermal state. This is important from practical view point, as most of the systems interact with 
environment and get thermalized eventually. We find the maximal amount of coherence that can be created from a thermal state at a given temperature and find a protocol to achieve this. Moreover, we  find the amount of coherence that can be created with limited available energy. Thus, our study establishes a link between coherence and thermodynamic resource theories and reveals the limitations 
imposed by thermodynamics on the processing of the coherence. Additionally, we have performed a comparative study between the coherence creation and total correlation creation with the same amount of energy at our disposal. We show that when one creates the maximum coherence with limited energy, the total correlation created in the process is always upper bounded by the amount of coherence created and vice versa. As correlation and coherence both are useful resources, processing them simultaneously is fruitful. However, our result shows that, at least in two qubit systems, there is no way to create the maximal coherence and correlation simultaneously via unitary transformations.
Recently, the importance of coherence in improving the performance of thermal machines has been explicitly established and the implications of coherence on the thermodynamic behavior of quantum systems have been studied. Therefore, it is justified to believe that the study of the thermodynamic cost and limitations of thermodynamic laws on the processing of quantum coherence can be far reaching. The results in this paper are a step forward in this direction. 

\smallskip
\noindent
\begin{acknowledgments}
A.M., U.S. and S.B.  acknowledge the research fellowship of Department of Atomic Energy, Government of India.
\end{acknowledgments}

\appendix
\setcounter{equation}{0}
\section{Coherence vs Correlation in two qubit systems}
\label{appendix-some-proof}
Here, for two qubit states with the Hamiltonian $H=H_A+H_B$, we show that maximum correlation and maximum coherence cannot be created simultaneously via a unitary transformation, starting from a thermal state with limited available energy. We consider both the cases where $E_A=E_B$ and 
$E_A\neq E_B$.
\subsection{\texorpdfstring{$\bf{E_A=E_B}$}{}}
For the case where the initial state is $\rho_T^{\otimes2}$ with $\rho_T=\rho^{A}_T=\rho^{B}_T=\mbox{diag}\{ p,1-p \}$ and the final state is in the $X-$state form, given by
\begin{align}
\rho_f=\left(\begin{array}{cccc}
             q^2 & 0 & 0 & Y\\
             0 & q(1-q) &X & 0\\
             0 & X^* & q(1-q) & 0\\
             Y^* & 0 & 0 & (1-q)^2
\end{array}\right).
\end{align}
Note that $p\geq q\geq 1/2 \geq (1-q)\geq (1-p)$. Here, $|Y| \leq q(1-q)$ and $|X| \leq q(1-q)$ so that, $\rho_f$ is positive semi-definite. Let $p=\frac{1}{2}+\epsilon$ and $q=\frac{1}{2}+\epsilon'$, where $\frac{1}{2}>\epsilon>\epsilon'>0$. The eigenvalues of this final density matrix are given by
\begin{eqnarray}
&&\lambda_{1,4}=\frac{1}{2}\left( q^2+(1-q)^2 \pm \sqrt{(q^2-(1-q)^2)^2+4|Y|^2} 
\right),\\
&&\lambda_{2,3}=~~~~q(1-q)\pm |X|.
\end{eqnarray}
As the unitary transformation preserves the eigenvalues, two of the eigenvalues  of the final density matrix must be equal to $p(1-p)$ and and the other two must be equal to $p^2$ and $(1-p)^2$ respectively.

\smallskip
\noindent 
{\bf Case 1:} Let us first assume $\lambda_2=\lambda_3=p(1-p)$. Then we find that $|X|=0$ and $q=p~ \mbox{or}~ q=1-p$. Since we know $p\geq1/2$, then $q\leq1/2$ for $q=1-p$. Hence, $q\neq1-p$. $q=p$ can only happen under identity operation. Therefore, $\lambda_2\neq \lambda_3$.

\smallskip
\noindent 
{\bf Case 2:} Assume $\lambda_1=\lambda_4=p(1-p)$, then we have
\begin{eqnarray}
\label{contra}
p(1-p)&&=\frac{q^2+(1-q)^2}{2}+\frac{q^2-(1-q)^2}{2}M\nonumber\\
&&=\frac{q^2+(1-q)^2}{2}-\frac{q^2-(1-q)^2}{2}M,
\end{eqnarray}
where $M=\sqrt{1+\frac{4|Y|^2}{(2q-1)^2}}$. From Eq. (\ref{contra}), we have $M=0$ which is a contradiction since $M\geq1$. Therefore, Eq. (\ref{contra}) cannot be satisfied. 

\smallskip
\noindent 
{\bf Case 3:} As $p^2\geq p(1-p) \geq (1-p)^2$, other two possibilities are $\lambda_1=\lambda_3=p(1-p)$ or $\lambda_4=\lambda_2=p(1-p)$. Note that we always have $\lambda_1>\lambda_3$. Therefore, the only possibility we have to check is $\lambda_4=\lambda_2=p(1-p)$. For that we have
\begin{eqnarray}
\lambda_2=p(1-p)&&\Rightarrow |X|=p(1-p)-q(1-q)\nonumber\\
&&\Rightarrow |X|=-(\epsilon^2-\epsilon^{'2}),
\end{eqnarray}
which is a contradiction as the R.H.S. is negative since $\epsilon>\epsilon'$. Therefore, it is also not possible.

\subsection{\texorpdfstring{$\bf{E_A \neq E_B}$}{}}
Let us relabel the diagonal entries of the initial density matrix as
\begin{align}
\rho_T^{AB}=a_1|00\rangle\langle 00|+a_2|01\rangle\langle 
01|+a_3|10\rangle\langle10|+a_4|11\rangle\langle11|.\nonumber
\end{align}
Here, $\{a_i\}$ is an arbitrary probability distribution that depends on the energy levels $E_A, E_B$ and the initial temperature $T$. We argue that the unitary transformations that map the initial state into an $X-$state starting from a two qubit thermal state at arbitrary finite temperature $T$, are only allowed to create correlation among the subspaces spanned by $\{|00\rangle, |11\rangle\}$ and $\{|01\rangle, 
|10\rangle\}$, separately, i.e., no correlation can be created between these two subspaces. Thus, the unitary transformation that maximizes the total correlation acts on the blocks spanned by $\{|00\rangle, |11\rangle\}$ and $\{|01\rangle, |10\rangle\}$, separately. Given this, again from comparing eigenvalues, it can be argued that total correlation and coherence cannot be maximized simultaneously by unitary transformations in two qubit systems when the Hamiltonian of the systems are not the same.

\bibliographystyle{apsrev4-1}
\bibliography{lit-therm-coh}

\end{document}